\documentclass[11pt]{article}
\usepackage{amsmath,amsthm,amssymb,epsfig,sectsty,caption,color,times}
\usepackage[letterpaper,hmargin=0.95in,vmargin=1.1in]{geometry}
\usepackage[dvipdfm,bookmarks=false,linkbordercolor={0 1 1},citebordercolor={0 1 1}]{hyperref}

\usepackage{algorithm, algpseudocode}
\algrenewcommand\algorithmicrequire{\textbf{\quad Input:}}
\algrenewcommand\algorithmicensure{\textbf{\quad Output:}}

\theoremstyle{plain}
\newtheorem{theorem}{Theorem}[section]
\newtheorem{lemma}[theorem]{Lemma}
\newtheorem{claim}[theorem]{Claim}

\theoremstyle{definition}

\newcommand{\qedsymb}{\hfill{\rule{2mm}{2mm}}}
\renewenvironment{proof}{\begin{trivlist} \item[\hspace{\labelsep}{\bf \noindent Proof.\/}] }{\qedsymb\end{trivlist}}%
\newenvironment{proofof}[1]{\begin{trivlist} \item[\hspace{\labelsep}{\bf \noindent Proof of #1.\/}] }{\qedsymb\end{trivlist}}%
\newcommand{\bbR}{\mathbb{R}}
\newcommand{\bbN}{\mathbb{N}}

\sectionfont{\large} \subsectionfont{\normalsize}

\begin{document}
\begin{titlepage}
\title{Efficient Submodular Function Maximization\\under Linear Packing Constraints}
\author{
    Yossi Azar\thanks{Blavatnik School of Computer Science, Tel-Aviv
    University, Israel. Email: \href{mailto:azar@tau.ac.il}
    {\tt azar@tau.ac.il}. Supported in part by the Israel Science
    Foundation (grant No. 1404/10).}
    \and
    Iftah Gamzu\thanks{Blavatnik School of Computer Science, Tel-Aviv
    University, Israel and Computer Science Division, The Open University, Israel.
    Email: \href{mailto:iftah.gamzu@cs.tau.ac.il}{\tt iftah.gamzu@cs.tau.ac.il}}
}
\date{}
\maketitle

\begin{abstract}
We study the problem of maximizing a monotone submodular set
function subject to linear packing constraints. An instance of
this problem consists of a matrix $A \in [0,1]^{m \times n}$, a
vector $b \in [1,\infty)^m$, and a monotone submodular set
function $f: 2^{[n]} \rightarrow \bbR_+$. The objective is to find
a set $S$ that maximizes $f(S)$ subject to $A x_{S} \leq b$, where
$x_S$ stands for the characteristic vector of the set $S$. A
well-studied special case of this problem is when $f$ is linear.
This special case captures the class of packing integer programs.

Our main contribution is an efficient combinatorial algorithm that
achieves an approximation ratio of $\Omega(1 / m^{1/W})$, where $W
= \min\{b_i / A_{ij} : A_{ij} > 0\}$ is the width of the packing
constraints. This result matches the best known performance
guarantee for the linear case. One immediate corollary of this
result is that the algorithm under consideration achieves constant
factor approximation when the number of constraints is constant or
when the width of the constraints is sufficiently large. This
motivates us to study the large width setting, trying to determine
its exact approximability. We develop an algorithm that has an
approximation ratio of $(1 - \epsilon)(1 - 1/e)$ when $W =
\Omega(\ln m / \epsilon^2)$. This result essentially matches the
theoretical lower bound of $1 - 1/e$. We also study the special
setting in which the matrix $A$ is binary and $k$-column sparse. A
$k$-column sparse matrix has at most $k$ non-zero entries in each
of its column. We design a fast combinatorial algorithm that
achieves an approximation ratio of $\Omega(1 / (Wk^{1/W}))$, that
is, its performance guarantee only depends on the sparsity and
width parameters.
\end{abstract}

\thispagestyle{empty}
\end{titlepage}

\section{Introduction} \label{sec:Intro} Let $f: 2^{[n]}
\rightarrow \bbR$ be a set function, where $[n] = \{1, 2, \ldots,
n\}$. The function $f$ is called \emph{submodular} if and only if
$f(S) + f(T) \geq f(S \cup T) + f(S \cap T)$, for all $S, T
\subseteq [n]$. An alternative definition of submodularity is
through the property of decreasing marginal values. Given a
function $f: 2^{[n]} \rightarrow \bbR$ and a set $S \subseteq
[n]$, the function $f_S$ is defined by $f_S(j) = f(S \cup \{j\}) -
f(S)$. The value $f_S(j)$ is called the incremental marginal value
of element $j$ to the set $S$. The \emph{decreasing marginal
values} property requires that $f_S(j)$ is non-increasing function
of $S$ for every fixed $j$. Formally, it requires that $f_S(j)
\geq f_T(j)$ for all $S \subseteq T$. Since the amount of
information necessary to convey an arbitrary submodular function
may be exponential, we assume a value oracle access to the
function. A \emph{value oracle} for the function $f$ allows us to
query about the value of $f(S)$ for any set $S$. Throughout the
rest of the paper, whenever we refer to a submodular function, we
shall also imply a \emph{normalized} and \emph{monotone} function.
Specifically, we assume that a submodular function $f$ also
satisfies $f(\emptyset) = 0$ and $f(S) \leq f(T)$ whenever $S
\subseteq T$.

In this paper, we focus our attention on the problem (or rather
class of problems) of maximizing a monotone submodular set
function subject to linear packing constraints. Formally, the
input of this problem consists of a matrix $A \in [0,1]^{m \times
n}$, a vector $b \in [1,\infty)^m$, and a monotone submodular set
function $f: 2^{[n]} \rightarrow \bbR_+$. The objective is to find
a set $S$ that maximizes $f(S)$ subject to $A x_{S} \leq b$, where
$x_S$ stands for the characteristic vector of the set $S$. We note
that the domain restrictions on the entries of $A$ and $b$ are
without loss of generality since arbitrary non-negative packing
constraints can be reduced to the above form by first eliminating
any element $j$ for which there is some constraint $i$ such that
$A_{ij} > b_i$, and then scaling the input (see, e.g., the
discussion in~\cite{Srinivasan99}). A well-studied special setting
of our problem is when the objective function $f$ is
\emph{linear}, namely, there is a weight vector $c \in \bbR_+^n$
such that $f(S) = \sum_{j \in S} c_j$. This special setting
captures the class of packing integer programs, which models many
fundamental combinatorial optimization problems, including maximum
independent set, hypergraph matching, and disjoint paths.

\medskip \noindent {\bf Previous work.}
Submodular functions play an instrumental role in computer
science, economics, and operations research as they form a rich
class that is general enough to be valuable for applications, but
still has plenty of structure to allow positive results. These
properties seem to make submodular functions a natural candidate
of choice for objective functions in optimization problems.
Indeed, over the last few years, there has been a surge of
interest in understanding the limits of tractability of
optimization problems in which the classic linear objective
function was replaced by a submodular one.

There has been a long line of research on maximizing monotone
submodular functions subject to matroid and knapsack constraints.
Arguably, the most classic scenario is maximizing a submodular
function subject to a cardinality constraint, that is, $\max
\{f(S) : |S| \leq k\}$. It is known that a simple greedy algorithm
achieves an approximation ratio of $1 - 1 /e$ for this
problem~\cite{NemhauserWF78}. Furthermore, this result is optimal
in two different ways: (i) given only oracle access to $f$, one
cannot attain a better approximation ratio without asking
exponentially many value queries~\cite{NemhauserW78}, and (ii)
even if $f$ has a compact representation, it is still NP-hard to
obtain a better approximation result~\cite{Feige98}. The greedy
approach and its variants has been shown to be useful in
additional constraint
structures~\cite{FisherNW78,KhullerMN99,ChekuriK04,GoundanS07}.
One relevant setting is maximizing a monotone submodular function
under a knapsack constraint~\cite{Wolsey82a}. A knapsack
constraint is essentially a single packing constraint, and may be
viewed as the weighted analog of a cardinality constraint.
Sviridenko~\cite{Sviridenko04} demonstrated that a greedy
algorithm with partial enumeration achieves an approximation
guarantee of $1 - 1/e$ for this problem.

Another approach that has been proven effective in handling
submodular function maximization under different constraint
structures is based on approximately solving a continuous
fractional relaxation of the problem, followed by pipage or
randomized rounding. The pipage rounding technique was originally
developed by Ageev and Sviridenko~\cite{AgeevS04}, and was adapted
to submodular maximization scenarios by Calinescu, Chekuri,
P{\'a}l and Vondr{\'a}k~\cite{CalinescuCPV07}.
Vondr{\'a}k~\cite{Vondrak08} utilized the continuous relaxation
approach to achieve a tight $(1 - 1/e)$-approximation for
maximizing a monotone submodular function subject to a matroid
constraint, and Kulik, Shachnai and Tamir~\cite{KulikST09} used
this approach to attain a $(1 - \epsilon)(1 - 1/e)$-approximation
for maximizing a monotone submodular function under a constant
number of packing constraints. Later on, Chekuri, Vondr{\'a}k and
Zenklusen~\cite{ChekuriVZ10} presented a dependent randomized
rounding scheme that can be utilized to extend those results for
maximizing a monotone submodular function subject to one matroid
and constant number of packing constraints. Recently, Feldman,
Naor and Schwartz~\cite{FeldmanNS11a} presented a new unified
continuous relaxation approach that finds approximate fractional
solutions in both monotone and non-monotone scenarios.

\medskip \noindent {\bf Our contribution.}
Our main result is an efficient multiplicative updates algorithm
for maximizing a monotone submodular function subject to any
number of linear packing constraints. The approximation ratio of
our algorithm matches the best known performance guarantee for the
special case when the objective function $f$ is linear, which is
achieved using the randomized rounding
technique~\cite{RaghavanT87,Raghavan88,Srinivasan99}. More
precisely, let $W = \min\{b_i / A_{ij} : A_{ij} > 0\}$ be the
\emph{width} of the packing constraints, we attain the following
result.

\begin{theorem} \label{th:MainResult1}
There is a deterministic polynomial-time algorithm that attains an
approximation guarantee of $\Omega(1 / m^{1/W})$ for maximizing a
monotone submodular function under linear packing constraints.
\end{theorem}

It is worth noting that our combinatorial algorithm is
deterministic and efficient. Moreover, our technique is different
than the two leading approaches used in the past for submodular
maximization, namely, the greedy approach and the continuous
relaxation approach. Our algorithm is based on a multiplicative
updates method (see,
e.g.,~\cite{PlotkinST95,Young95,GargK07,AzarR06,BriestKV05}). This
method is known to be fruitful for approximately solving problems
that can be cast as linear and integer programs. Nevertheless, the
analysis of these algorithms relies heavily on primal-dual
results, which are not applicable in our submodular setting. We
believe that this new approach may be suitable for other
submodular optimization problems. We also like to remark that a
comparable approximation guarantee may be obtained using the
continuous relaxation approach applied with randomized
rounding~\cite{ChekuriV10}. However, in contrast with that
approach, our algorithm is deterministic, efficient and
combinatorial.

One immediate corollary of Theorem~\ref{th:MainResult1} is that
the algorithm under consideration achieves a constant factor
approximation when the number of constraints is constant or when
the width of the packing constraints is sufficiently large, say $W
= \Omega(\ln m)$. This motivates us to study the large width
setting, trying to determine its exact approximability. The
following theorem summarizes our result in this context.

\begin{theorem} \label{th:MainResult2}
There is a deterministic polynomial-time algorithm that achieves
an approximation guarantee of $(1 - \epsilon)(1 - 1/e)$ for
maximizing a monotone submodular function subject to linear
packing constraints when $W = \Omega(\ln m / \epsilon^2)$, for any
fixed $\epsilon > 0$.
\end{theorem}

We note that this result almost matches the theoretical lower
bound of $1 - 1/e$, which already holds for maximizing a monotone
submodular function subject to a cardinality
constraint~\cite{NemhauserWF78,Feige98}. Specifically, the large
width setting captures the hard instances of that problem. We
remark that the $(1 - 1/e)$-approximation in the submodular
setting stands in contrast with a $(1 + \epsilon)$-approximation
which can be achieved by randomized rounding when the objective
function is linear and the width is sufficiently large.

We also study the interesting special setting of the problem in
which the constraints matrix is binary, namely, $A \in \{0,1\}^{m
\times n}$ instead of $A \in [0,1]^{m \times n}$. We demonstrate
how to fine-tune our algorithm and its analysis to achieve an
improved approximation guarantee of $\Omega(1 / m^{1/(W+1)})$.
This result is formalized in Theorem~\ref{th:MainResult3}. We like
to emphasize that this result is optimal unless $\mathrm{P} =
\mathrm{ZPP}$. Recently, Bansal et al.~\cite{BansalKNS10}
considered the special case of maximizing a submodular function
under $k$\emph{-column sparse} packing constraints. In this
setting, the constraints matrix has at most $k$ non-zero entries
in each column. They developed an algorithm whose approximation
ratio only depends on the sparsity and width parameters of the
input matrix. Specifically, they presented a $\Omega(1 /
k^{1/W})$-approximation algorithm that employs the continuous
relaxation approach in conjunction with randomized rounding and
alteration. We make a first step towards attaining their
performance guarantee in a deterministic and efficient way. We
present a fast combinatorial algorithm for the binary $k$-column
sparse setting whose approximation ratio only depends on the
sparsity and width parameters of the input matrix. The following
theorem outlines this result.

\begin{theorem} \label{th:MainResult4}
There is a deterministic polynomial-time algorithm that achieves
an approximation guarantee of $\Omega(1 / (W k^{1/W}))$ for
maximizing a monotone submodular function under binary packing
constraints.
\end{theorem}

\medskip \noindent {\bf Other related work.}
The problem of maximizing a non-monotone submodular function
without any structural constraints is known to be both NP-hard and
APX-hard since it generalizes the maximum cut problem. Feige,
Mirrokni and Vondr{\'a}k~\cite{FeigeMV07} developed an algorithm
whose approximation ratio is $0.4$. This result was iteratively
improved by Oveis Gharan and Vondr{\'a}k~\cite{GharanV11}, and
then by Feldman, Naor and Shwartz~\cite{FeldmanNS11b} to a ratio
of $0.42$. Lee, Mirrokni, Nagarajan and Sviridenko~\cite{LeeMNS10}
presented a $(1/4-\epsilon)$-approximation algorithm for
non-monotone submodular maximization subject to a constant number
of packing constraints. This result was iteratively improved by
Chekuri, Vondr{\'a}k and Zenklusen~\cite{ChekuriVZ11}, and then by
Feldman, Naor and Shwartz~\cite{FeldmanNS11a} to a ratio of
$1/e-\epsilon$. Vondr{\'a}k~\cite{Vondrak09}, and very recently,
Dobzinski and Vondr{\'a}k~\cite{DobzinskiV12} developed general
approaches to derive inapproximability results in the value oracle
model.

Unlike submodular function maximization, the problem of minimizing
a submodular function can be performed efficiently, either by the
ellipsoid algorithm~\cite{GrotschelLS81} or through strongly
polynomial-time combinatorial
algorithms~\cite{Schrijver00,IwataFF01,Iwata03,Orlin07,Iwata08,IwataO09}.
Goemans, Harvey, Iwata and Mirrokni~\cite{GoemansHIM09} considered
the problem of explicitly constructing a function that
approximates a monotone submodular function while making a
polynomial number of oracle queries. They showed an essentially
tight $\tilde{O}(n^{1/2})$-approximate solution. Recently, several
submodular analogues of classical combinatorial optimization
problems have been studied~\cite{SvitkinaF08,GoelKTW09,IwataN09}.
These submodular problems are commonly considerably harder to
approximate than their linear counterparts. For example, the
minimum spanning tree problem, which is polynomial-time solvable
with linear cost functions is $\Omega(n)$-hard to approximate with
submodular cost functions~\cite{GoelKTW09}.

\section{Submodular Maximization with Linear Packing Constraints} \label{sec:GeneralPacking}
In this section, we develop a multiplicative updates algorithm for
the problem and analyze its performance. An important input
parameter of our algorithmic template is an update factor. This
parameter plays an essential role in achieving the desired
approximation guarantees in the two settings of interest. We first
consider the general problem, and demonstrate that there is an
update factor for which our algorithm attains an approximation
ratio of $\Omega(1 / m^{1/W})$. In particular, this implies that
the algorithm achieves constant factor approximation for input
instances that have a large width, e.g., instances with $W =
\Omega(\ln m)$. This motivates us to study this large width
setting, trying to determine its exact approximability. We match
(up to a disparity of~$\epsilon$) the theoretical lower bound of
$1 - 1/e$ using a different update factor and a refined analysis.

\subsection{The algorithm} \label{subsec:Algorithm}
The multiplicative updates algorithm, formally described below,
maintains a collection of weights that are updated in a
multiplicative way. Informally, these weights capture the extent
to which each constraint is close to be violated under a given
solution. The algorithm is built around one main loop. In each
iteration of that loop, the algorithm extends the current solution
with a non-selected element that minimizes a normalized sum of the
weights. When the loop terminates, the algorithm returns the
resulting solution in case it is feasible; otherwise, either the
last selected element or the resulting solution without that
element is returned, depending on their value. Recall that $f_S(j)
= f(S \cup \{j\}) - f(S)$ is the incremental marginal value of
element $j$ to the set $S$, and $x_S$ is the characteristic vector
of the set $S$.

\begin{algorithm}
\caption{Multiplicative Updates}\label{cap:MultiplicativeUpdates}%
\begin{algorithmic}[1]
\Require A collection of linear packing constraints defined by $A \in [0,1]^{m \times n}$ and $b \in [1,\infty)^m$ %
\Statex \qquad\; A monotone submodular set function $f: 2^{[n]} \rightarrow \bbR_{+}$ %
\Statex \qquad\; An update factor $\lambda \in \bbR_+$ %
\Ensure A subset of $[n]$ \smallskip %

\State $S \leftarrow \emptyset$ %
\State \textbf{for} $i \leftarrow 1$ to $m$ \textbf{do} $w_i \leftarrow 1 / b_i$ \textbf{end for} \smallskip%

\While{$\sum_{i=1}^m b_i w_i \leq \lambda$ and $S \neq [n]$} \label{alg:StopCond} %
    \State Let $j \in [n] \setminus S$ be the element with minimal $\sum_{i=1}^{m} A_{ij}w_i / f_S(j)$ %
    \State $S \leftarrow S \cup \{j\}$ %
    \State \textbf{for} $i \leftarrow 1$ to $m$ \textbf{do} $w_i \leftarrow w_i \lambda^{A_{ij}/b_i}$ \textbf{end for} \label{alg:WeightUpdate} %
\EndWhile \smallskip %

\State \textbf{if} $Ax_S \leq b$ \textbf{then} return $S$ %
\State \textbf{else if} $f(S \setminus \{j\}) \geq f(\{j\})$ \textbf{then} return $S \setminus \{j\}$ %
\State \textbf{else} return $\{j\}$ \textbf{end if} %
\end{algorithmic}
\end{algorithm}

\subsection{Analysis} \label{subsec:Analysis}
In the remainder of this section, we analyze the performance of
the algorithm. We begin by establishing several lemmas that hold
independently of the value of the update factor. Later on, we
consider specific update factors, and study their effect on the
approximation ratio of the algorithm. For ease of presentation, it
would be convenient to first introduce some notation and
terminology:
\begin{itemize}
\item Let $S^* \subseteq [n]$ be a solution that maximizes the
submodular function subject to the linear packing constraints,
with value of $f(S^*)$.

\item Let $S_t$ be the solution at the end of iteration $t$ of the
algorithm, and note that $S_0 = \emptyset$ indicates the solution
at the beginning of the algorithm. Moreover, let $\gamma(t)$
denote the element selected at iteration $t$ of the algorithm, and
let $\delta_t = f(S_t) - f(S_{t-1})$ be its incremental marginal
value to the solution. Finally, let $w_{it}$ be the value of $w_i$
at the end of iteration $t$ of the algorithm, and remark that
$w_{i0} = 1 / b_i$ is the value of $w_i$ at the beginning of the
algorithm.

\item Let $\Lambda_t = \sum_{i = 1}^m b_i w_{it}$ and $\alpha_t =
\sum_{i=1}^{m} A_{i\gamma(t)}w_{i(t-1)} / \delta_{t}$. Notice that
the algorithm may proceed to iteration $t+1$ only if $\Lambda_t
\leq \lambda$, and that $\Lambda_0 = m$. Also note that $\alpha_t$
is the value which gave rise to the selection of element
$\gamma(t)$ at iteration $t$ of the algorithm.
\end{itemize}

\smallskip \noindent {\bf Correctness.}
We prove that the algorithm outputs a feasible solution. This is
achieved by demonstrating that the returned solution respects the
packing constraints.

\begin{lemma} \label{lemma:Feasibility}
The algorithm outputs a feasible solution.
\end{lemma}
\begin{proof}
Let us focus on the solution $S$ when the main loop terminates.
Clearly, if $S$ respects the packing constraints then the returned
solution also respects them. Thus, let us consider the case that
$S$ is infeasible. We next argue that $S$ became infeasible only
at the last iteration of the loop in which element $\ell$ was
selected. Consequently, by inspecting the last two lines of the
algorithm, one can conclude that the returned solution must be
feasible as it is either $S \setminus \{\ell\}$ or $\{\ell$\}.

For the purpose of establishing the previously mentioned argument,
let $\ell$ be the first element that induces a violation in some
constraint. Specifically, suppose $\ell$ induces a violation in
constraint $i$ at iteration $t$. Accordingly, $\sum_{j \in S_t}
A_{ij} > b_i$, and theretofore,
$$
b_i w_{it} = b_i w_{i0} \prod_{j \in S_t} \lambda^{A_{ij} / b_i} =
\lambda^{\sum_{j \in S_t} A_{ij} / b_i}
> \lambda \ ,
$$
where the last equality is due to the fact that $w_{i0} = 1 /
b_i$. This implies that $\Lambda_t > \lambda$, and hence, by
inspecting the main loop stopping condition, we know that the loop
must have terminated immediately after element $\ell$ was
selected.~
\end{proof}

\smallskip \noindent {\bf Approximation.}
We turn to analyze the approximation guarantee of the algorithm.
We begin by establishing a generic algebraic bound applicable for
any monotone submodular function and any arbitrary sequence of
element additions.

\begin{claim} \label{claim:SubmodularBound}
Given a submodular function $f: 2^{[n]} \to \bbR_+$, a set
collection $S_0 \subseteq S_1 \subseteq \cdots \subseteq S_t
\subseteq [n]$, and a set $S^* \subseteq [n]$ satisfying $f(S^*) >
f(S_t)$ then
$$
\sum_{\ell = 1}^{t} \frac{f(S_{\ell}) - f(S_{\ell - 1})}{f(S^*) -
f(S_{\ell - 1})} \leq \ln\left(\frac{f(S^*) - f(S_0)}{f(S^*) -
f(S_t)} \right) \ .
$$
\end{claim}
\begin{proof}
One should observe that for any $\ell = 1, \ldots, t$,
$$
\frac{f(S_{\ell}) - f(S_{\ell - 1})}{f(S^*) - f(S_{\ell - 1})} =
\int_{f(S_{\ell - 1})}^{f(S_{\ell})} \frac{1}{f(S^*) - f(S_{\ell -
1})}dx \leq \int_{f(S_{\ell - 1})}^{f(S_{\ell})} \frac{1}{f(S^*) -
x}dx \ ,
$$
where the inequality follows by noticing that the function $1 /
(f(S^*) - x)$ is monotonically increasing for $x \in [0, f(S^*))$.
As a consequence, we obtain that
$$
\sum_{\ell = 1}^{t} \frac{f(S_{\ell}) - f(S_{\ell - 1})}{f(S^*) -
f(S_{\ell - 1})} \leq \sum_{\ell = 1}^{t} \int_{f(S_{\ell -
1})}^{f(S_{\ell})} \frac{1}{f(S^*) - x}dx = \int_{f(S_0)}^{f(S_t)}
\frac{1}{f(S^*) - x}dx =
\ln\left(\frac{f(S^*)-f(S_0)}{f(S^*)-f(S_t)}\right) \ .
$$~
\end{proof}

We continue by bounding the value of the optimal solution using
the main parameters of the algorithm at the end of iteration
$\ell$.

\begin{claim} \label{claim:OPTBound}
$f(S^*) \leq f(S_\ell) + \Lambda_\ell / \alpha_{\ell+1}$ in every
iteration $\ell$.
\end{claim}
\begin{proof}
We know that the element selected at iteration $\ell + 1$
minimizes the term $\sum_{i=1}^{m} A_{ij}w_{i\ell} /
f_{S_\ell}(j)$ with respect to every $j \in [n] \setminus S_\ell$.
This clearly implies that $\alpha_{\ell + 1} \leq \sum_{i=1}^{m}
A_{ij}w_{i\ell} / f_{S_\ell}(j)$ for every $j$ under
consideration. Rearranging the terms in this inequality, we can
bound the marginal value of each element $j \in [n] \setminus
S_\ell$ with respect to $S_\ell$ as
$$
f_{S_\ell}(j) \leq \sum_{i=1}^{m}
\frac{A_{ij}w_{i\ell}}{\alpha_{\ell+1}} \ .
$$
Let $J^* = \{j: j \in S^* \text{ and } j \notin S_\ell \}$ be the
set of elements selected by the optimal solution, but not selected
by the algorithm up to the end of iteration $\ell$. Note that $J^*
\subseteq [n] \setminus S_\ell$, and notice that
$$
f(S^*) \leq f(S^* \cup S_\ell) \leq f(S_\ell) + \sum_{j \in J^*}
f_{S_\ell}(j) \ ,
$$
where the first inequality follows from the monotonicity of $f$,
and the last inequality holds as a result of its submodularity.
Specifically, the latter inequality is obtained using the
decreasing marginal values property. We now focus on bounding the
above right-hand side term. For this purpose, we utilize the bound
derived earlier on the marginal values of the elements in $[n]
\setminus S_\ell$, and attain
$$
\sum_{j \in J^*} f_{S_\ell}(j) \leq \sum_{j \in J^*}
\sum_{i=1}^{m} \frac{A_{ij}w_{i\ell}}{\alpha_{\ell+1}} =
\sum_{i=1}^{m} \frac{w_{i\ell}}{\alpha_{\ell+1}} \sum_{j \in J^*}
A_{ij} \leq \sum_{i=1}^{m} \frac{b_i w_{i\ell}}{\alpha_{\ell+1}} =
\frac{\Lambda_\ell}{\alpha_{\ell+1}} \ ,
$$
where the last inequality follows by recalling that the elements
in $J^*$ are a subset of the elements in the optimal solution, and
thus, constitute a feasible solution respecting all constraints.
As a result, $\sum_{j \in J^*} A_{ij} \leq b_i$.~
\end{proof}

We next demonstrate that the algorithm attains an approximation
guarantee of $\Omega(1 / m^{1/W})$ when the update factor is
$\lambda = e^W m$. Recall that $W = \min\{b_i / A_{ij} : A_{ij} >
0\}$ is the width of the constraints.

\begin{lemma} \label{lemma:Approx1}
The algorithm archives $\Omega(1 / m^{1/W})$-approximation by
using $\lambda = e^W m$.
\end{lemma}
\begin{proof}
Suppose the main loop terminates after $t$ iterations. Notice that
when the loop terminates either $S_t = [n]$ or $\sum_{i=1}^m b_i
w_{it} > e^W m$. In the former case, one can easily infer that the
returned solution is $1/2$-approximation to the optimal solution.
Specifically, if $S_t$ is returned by the algorithm then the
outcome is clearly optimal since $S_t$ consists of all elements,
and if one of $S_t \setminus \{j\}$ or $\{j\}$ is returned then
the value of the solution is a $1/2$-approximation since
$$
\max\big\{f(S_t \setminus \{j\}), f(\{j\})\big\} \geq
\frac{1}{2}\big(f(S_t \setminus \{j\}) + f(\{j\})\big) \geq
\frac{1}{2}f(S_t) \ ,
$$
where the last inequality uses the submodularity of $f$. In fact,
one can easily validate that the above analysis also holds in case
that $f(S_t) \geq f(S^*)$, which can happen since $S_t$ may be
infeasible. Hence, in the remainder of the proof, we shall assume
that $f(S^*) > f(S_t)$ and that the loop terminates with
$\Lambda_t = \sum_{i=1}^m b_i w_{it} > e^W m$.

We concentrate on upper bounding the value of $\Lambda_t$. For
this purpose, we analyze the change in $\sum_{i=1}^m b_i w_i$
along the loop iterations. Observe that for any $\ell = 1, \ldots,
t$,
\begin{eqnarray*}
\Lambda_{\ell} = \sum_{i=1}^m b_i w_{i\ell} & = & \sum_{i=1}^m b_i w_{i(\ell-1)} \cdot \left(e^W m\right)^{A_{i\gamma(\ell)}/b_i}\\
& \leq & \sum_{i=1}^m b_i w_{i(\ell-1)} \cdot \left(1 + \frac{e W m^{1/W} A_{i\gamma(\ell)}}{b_i}\right)\\
& = & \sum_{i=1}^m b_i w_{i(\ell-1)} + e W m^{1/W} \sum_{i=1}^m A_{i\gamma(\ell)} w_{i(\ell-1)}\\
& = & \Lambda_{\ell - 1} + e W m^{1/W} \alpha_\ell \delta_\ell \ .
\end{eqnarray*}
The first inequality follows by plugging $a = e m^{1/W}$ and $y =
W A_{i\gamma(\ell)}/ b_i$ to the inequality $a^y \leq 1 + ay$,
which is known to be valid for any $a \in \bbR_+$ and $y \in
[0,1]$, and the last equality results from the definition of
$\alpha_\ell$. By Claim~\ref{claim:OPTBound}, we know that
$\alpha_\ell \leq \Lambda_{\ell-1} / (f(S^*) - f(S_{\ell-1}))$ in
case $f(S^*) > f(S_{\ell-1})$. The latter condition clearly holds
since $f(S^*) > f(S_t)$ by previous assumption, and $f(S_t) \geq
f(S_{\ell-1})$ for any $\ell$ under consideration. Therefore,
$$
\Lambda_\ell \leq \Lambda_{\ell-1} \cdot \left(1 + \frac{e W
m^{1/W} \delta_\ell}{f(S^*) - f(S_{\ell-1})}\right) \leq
\Lambda_{\ell-1} \cdot \exp\left(\frac{e W m^{1/W}
\delta_\ell}{f(S^*) - f(S_{\ell-1})}\right) \ ,
$$
where the last inequality is due to the fact that $1 + y \leq
e^y$. The resulting recursive definition can be used, in
conjunction with the base case $\Lambda_0 = m$, to upper bound
$\Lambda_t$ by
$$
\Lambda_t \leq \Lambda_0 \cdot \prod_{\ell = 1}^{t} \exp
\left(\frac{e W m^{1/W} \delta_\ell}{f(S^*) -
f(S_{\ell-1})}\right) = m \cdot \exp \left(e W m^{1/W} \sum_{\ell
= 1}^{t}\frac{ f(S_\ell) - f(S_{\ell-1})}{f(S^*) -
f(S_{\ell-1})}\right) \ .
$$
Recall that we assumed that the loop terminated with $\Lambda_t
> e^W m$. This lower bound on $\Lambda_t$ can be utilized, together
with the upper bound on $\Lambda_t$, to yield
$$
1 \leq em^{1/W} \sum_{\ell = 1}^{t}\frac{ f(S_\ell) -
f(S_{\ell-1})}{f(S^*) - f(S_{\ell-1})} \leq em^{1/W}
\ln\left(\frac{f(S^*) - f(S_0)}{f(S^*) - f(S_t)} \right) \ ,
$$
where the last inequality is due to the
Claim~\ref{claim:SubmodularBound}. We note that $f(S_0) = 0$ since
$f$ is normalized and $S_0 = \emptyset$. Subsequently, one can
obtain that $1 - 1 / \exp(1 / em^{1/W}) \leq f(S_t) / f(S^*)$
using simple algebraic manipulations. This can be further
simplified to $1 / (em^{1/W} + 1) \leq f(S_t) / f(S^*)$ by
reutilizing the fact that $1 + y \leq e^y$. Notice that this
proves that the algorithm archives $\Omega(1 /
m^{1/W})$-approximation since the value of the returned solution
is at least $f(S_t) / 2$. This follows from arguments similar to
those presented at the beginning of the proof.~
\end{proof}

We are now ready to complete the proof of the first main result of
the paper. We note that this result matches the best known
approximation guarantee for the case that the objective function
$f$ is linear, achievable using the randomized rounding
technique~\cite{RaghavanT87,Raghavan88,Srinivasan99}.

\begin{proofof}{Theorem~\ref{th:MainResult1}}
By Lemma~\ref{lemma:Feasibility} and Lemma~\ref{lemma:Approx1}, we
know that when the algorithm uses an update factor of $\lambda =
e^W m$, it constructs a feasible solution which approximates the
optimal solution within a factor of $\Omega(1 / m^{1/W})$.~
\end{proofof}

One immediate corollary of this theorem is that the algorithm
under consideration attains a constant approximation guarantee
when the number of constraints is constant or when the width is
sufficiently large, say $W = \Omega(\ln m)$. In particular, one
can reexamine the analysis presented in the proof of
Lemma~\ref{lemma:Approx1}, and deduce that the approximation ratio
of the algorithm approaches $1 / (2e + 2)$ for sufficiently large
$W$'s. A natural followup question is whether one can improve upon
this result. In what follows, we demonstrate that we can beat this
approximation ratio by a careful selection of the update factor.
We present a refined analysis that proves an approximation ratio
of $(1 - \epsilon)(1 - 1/e)$ when $W = \Omega(\ln m /
\epsilon^2)$. In particular, our analysis avoids the two-factor
loss due to the max-selection in the last two lines of the
algorithm.

\begin{lemma} \label{lemma:Approx2}
The algorithm achieves an approximation ratio of $(1 -
4\epsilon)(1 - 1/e)$ by using $\lambda = e^{\epsilon W}$ when $W
\geq \max\{\ln m / \epsilon^2, 1 / \epsilon\}$ for any fixed
$\epsilon > 0$.
\end{lemma}
\begin{proof}
Suppose the main loop terminates after $t+1$ iterations. Let us
consider the case that it terminates with $\sum_{i=1}^m b_i
w_{i(t+1)} < e^{\epsilon W}$. Note that this implies that $S_{t+1}
= [n]$. One can also argue that $S_{t+1}$ is the returned solution
since it is feasible. The feasibility of $S_{t+1}$ follows from
arguments similar to those presented in the proof of
Lemma~\ref{lemma:Feasibility}. Specifically, one can demonstrate
that if $S_{t+1}$ violates some constraint $i$ then $b_i
w_{i(t+1)} > e^{\epsilon W}$. Obviously, the returned solution is
optimal as $S_{t+1}$ consists of all elements. Hence, in the
remainder of the proof, we shall focus on the case that the loop
terminates with $\sum_{i=1}^m b_i w_{i(t+1)} \geq e^{\epsilon W}$.

We next argue that solution constructed up to and not including
the last iteration, namely $S_t$, achieves the claimed
approximation guarantee. Note that this implies that the returned
solution must also have the desired performance guarantee since
$S_t$ is feasible. The feasibility of $S_t$ also follows from
arguments similar to those exhibited in the proof of
Lemma~\ref{lemma:Feasibility}. Specifically, one can establish
that if $S_{t+1}$ is infeasible then it became infeasible only at
the last iteration of the loop, and thus, $S_t$ is feasible. We
turn to bound the value of $\Lambda_t$. A lower bound can be
easily obtained by noticing that
$$
\Lambda_t e^{\epsilon} = \sum_{i=1}^m b_i w_{it} \cdot
\left(e^{\epsilon W}\right)^{1/W} \geq \sum_{i=1}^m b_i w_{i(t+1)}
\geq e^{\epsilon W} \ ,
$$
and therefore, $\Lambda_t \geq e^{\epsilon (W-1)}$. Similarly to
the proof of Lemma~\ref{lemma:Approx1}, we derive an upper bound
on $\Lambda_t$ by analyzing the change in $\sum_{i=1}^m b_i w_i$
along the loop iterations. Observe that for any $\ell = 1, \ldots,
t$,
\begin{eqnarray*}
\Lambda_{\ell} = \sum_{i=1}^m b_i w_{i\ell} & = & \sum_{i=1}^m b_i w_{i(\ell-1)} \cdot \left(e^{\epsilon W}\right)^{A_{i\gamma(\ell)}/b_i}\\
& \leq & \sum_{i=1}^m b_i w_{i(\ell-1)} \cdot \left(1 + \frac{\epsilon W A_{i\gamma(\ell)}}{b_i} + \left(\frac{\epsilon W A_{i\gamma(\ell)}}{b_i}\right)^2\right)\\
& \leq & \sum_{i=1}^m b_i w_{i(\ell-1)} +  (\epsilon W + \epsilon^2 W) \sum_{i=1}^m A_{i\gamma(\ell)} w_{i(\ell-1)}\\
& = & \Lambda_{\ell - 1} + (\epsilon W + \epsilon^2 W) \alpha_\ell
\delta_\ell \ .
\end{eqnarray*}
The first inequality follows from the fact that $e^y \leq 1 + y +
y^2$ for any $y \in [0,1]$, which can be derived from the
corresponding Taylor expansion. The last inequality is obtained by
using the fact that $W A_{i\gamma(\ell)} / b_i \leq 1$ to reason
that $(\epsilon W A_{i\gamma(\ell)}/ b_i)^2 \leq \epsilon^2 W
A_{i\gamma(\ell)} / b_i$. Finally, the last equality results from
the definition of $\alpha_\ell$. By Claim~\ref{claim:OPTBound}, we
know that $\alpha_\ell \leq \Lambda_{\ell-1} / (f(S^*) -
f(S_{\ell-1}))$ when $f(S^*) > f(S_{\ell-1})$. The latter
condition clearly holds since $f(S^*) \geq f(S_t)$ as $S_t$ is a
feasible solution, and $f(S_t) \geq f(S_{\ell-1})$ for any $\ell$
under consideration. Therefore,
$$
\Lambda_\ell \leq \Lambda_{\ell-1} \cdot \left(1 + \frac{(\epsilon
W + \epsilon^2 W) \delta_\ell}{f(S^*) - f(S_{\ell-1})}\right) \leq
\Lambda_{\ell-1} \cdot \exp\left(\frac{(\epsilon W + \epsilon^2 W)
\delta_\ell}{f(S^*) - f(S_{\ell-1})}\right) \ ,
$$
where the last inequality is due to the fact that $1 + y \leq
e^y$. The resulting recursive definition can be used to upper
bound $\Lambda_t$ by
$$
\Lambda_t \leq \Lambda_0 \cdot \prod_{\ell = 1}^{t} \exp
\left(\frac{(\epsilon W + \epsilon^2 W) \delta_\ell}{f(S^*) -
f(S_{\ell-1})}\right) \leq \exp \left(\epsilon^2 W + (\epsilon W +
\epsilon^2 W) \sum_{\ell = 1}^{t}\frac{ f(S_\ell) -
f(S_{\ell-1})}{f(S^*) - f(S_{\ell-1})}\right) \ ,
$$
where the last inequality holds since $\Lambda_0 = m \leq
\exp(\epsilon^2 W)$ by our assumption regarding the width of the
constraints. Recall that we previously demonstrated that
$\Lambda_t \geq \exp(\epsilon (W-1))$. This lower bound on
$\Lambda_t$ can be utilized, together with the upper bound on
$\Lambda_t$, to yield
$$
\frac{\epsilon (W-1) - \epsilon^2 W}{\epsilon W + \epsilon^2 W}
\leq \sum_{\ell = 1}^{t}\frac{ f(S_\ell) - f(S_{\ell-1})}{f(S^*) -
f(S_{\ell-1})} \leq \ln\left(\frac{f(S^*) - f(S_0)}{f(S^*) -
f(S_t)} \right) \ ,
$$
where the last inequality is due to the
Claim~\ref{claim:SubmodularBound}. Note that $f(S_0) = 0$ as $f$
is normalized and $S_0 = \emptyset$. Also notice that $(\epsilon
(W-1) - \epsilon^2 W)/(\epsilon W + \epsilon^2 W) \geq (1 -
2\epsilon) / (1 + \epsilon) \geq 1 - 3\epsilon$. Subsequently, one
can obtain that $1 - 1 / \exp(1  - 3 \epsilon) \leq f(S_t) /
f(S^*)$ using simple algebraic manipulations. The claimed
approximation ratio follows by noticing that
$$
1 - \frac{e^{3 \epsilon}}{e} \geq 1 - \frac{1 + 3 \epsilon + 9
\epsilon^2}{e} \geq (1 - 4\epsilon)\left(1 - \frac{1}{e}\right) \
,
$$
where the first inequality reuses the fact that $e^y \leq 1 + y +
y^2$ for any $y \in [0,1]$, and both inequalities assume that
$\epsilon \leq 1/4$, which is the interesting range of values for
$\epsilon$.~
\end{proof}

We are now ready to complete the proof of the second principal
result of the paper. We note that this result almost matches the
theoretical lower bound of $1 - 1/e$, which already holds for
maximizing a monotone submodular function subject to a cardinality
constraint~\cite{NemhauserWF78,Feige98}. In particular, our large
width setting captures the hard instances of the latter problem as
this problem can be solved in polynomial-time when $W = O(1 /
\epsilon)$ by enumerating over all sets of size at most $W$.

\begin{proofof}{Theorem~\ref{th:MainResult2}}
Given an instance of the problem in which $W = \Omega(\ln m /
\epsilon^2)$ for any fixed $\epsilon > 0$,
Lemma~\ref{lemma:Feasibility} and Lemma~\ref{lemma:Approx2}
guarantee that employing the algorithm with an update factor of
$\lambda = e^{\epsilon W/4}$ results in a feasible solution that
approximates the optimal solution within a factor of $(1 -
\epsilon)(1 - 1/e)$.~
\end{proofof}

\section{Submodular Maximization with Binary Packing Constraints} \label{subsec:0-1Packing}
In this section, we consider the special setting of monotone
submodular maximization under binary packing constraints, namely,
when $A \in \{0,1\}^{m \times n}$ instead of $A \in [0,1]^{m
\times n}$. Note that we may assume without loss of generality
that $b \in \bbN_+^m$ since each vector entry can be rounded down
to the nearest integer without any consequences whatsoever. This
natural setting has been considered in the past for linear
objective functions. Similarly to the general linear case, the
randomized rounding technique attains the best known approximation
guarantee in this case as well. In particular, it achieves an
approximation ratio of $\Omega(1 / m^{1/(W+1)})$, which is
polynomially better than the general case. This outcome is also
known to be optimal unless $\mathrm{P} =
\mathrm{ZPP}$~\cite{ChekuriK04}. We can demonstrate that our
multiplicative updates approach from
Section~\ref{sec:GeneralPacking} can be utilized to obtain the
above-mentioned improved approximation guarantee for the
underlying setting. This requires a fine-tuning of the algorithm
and its analysis. We defer these technical details to
Appendix~\ref{appsec:0-1}.

We next develop a different multiplicative updates algorithm for
the special setting in which the constraints matrix is $k$-column
sparse. In this case, the number of $1$-value entries in each
column of the input matrix is at most $k$. We prove that our
algorithm achieves an approximation guarantee that does not depend
on the number of rows $m$, but only depends on the sparsity
parameter $k$ and width parameter $W$. More precisely, we
establish that the algorithm attains an approximation ratio of
$\Omega(1 / (Wk^{1/W}))$.

\subsection{The algorithm} \label{subsec:SparseAlgorithm}
The multiplicative updates algorithm, formally described below,
maintains a collection of weights that capture the extent to which
each constraint is close to be violated under a given solution.
The algorithm is built around one main loop. In each iteration of
that loop, the algorithm considers a remaining element whose
marginal contribution to the current solution is maximal, and adds
it to the solution set if its corresponding sum of weights is
sufficiently small. In any case, the element under consideration
is removed from the list of remaining elements. When the loop
terminates, the algorithm returns the resulting solution. Recall
that $f_S(j) = f(S \cup \{j\}) - f(S)$ is the incremental marginal
value of element $j$ to the set $S$

\begin{algorithm}
\caption{Column Sparse Multiplicative Updates}\label{cap:SparseMultiplicativeUpdates}%
\begin{algorithmic}[1]
\Require A collection of linear packing constraints defined by $A \in \{0,1\}^{m \times n}$ and $b \in \bbN_+^m$ %
\Statex \qquad\; A monotone submodular set function $f: 2^{[n]} \rightarrow \bbR_{+}$ %
\Statex \qquad\; An update factor $\lambda \in \bbR_+$ %
\Ensure A subset of $[n]$ \smallskip %

\State $S \leftarrow \emptyset$, $R \leftarrow [n]$ %
\State \textbf{for} $i \leftarrow 1$ to $m$ \textbf{do} $w_i \leftarrow 0$ \textbf{end for} \smallskip%

\While{$R \neq \emptyset$} \label{alg:SparseStopCond} %
    \State Let $j \in R$ be the element with maximal $f_S(j)$ %
    \State \textbf{if} $\sum_{i=1}^m A_{ij} w_i < (\lambda - 1)$ \textbf{then} $S \leftarrow S \cup \{j\}$ %
    \State $R \leftarrow R \setminus \{j\}$ %
    \State \textbf{for} $i \leftarrow 1$ to $m$ \textbf{do} $w_i \leftarrow \lambda^{\sum_{j \in S} A_{ij} / b_i} - 1$ \textbf{end for} \label{alg:SparseWeightUpdate} %
\EndWhile \smallskip %

\State return $S$ %
\end{algorithmic}
\end{algorithm}

\subsection{Analysis} \label{subsec:Analysis}
In what follows, we analyze the performance of the algorithm. We
begin by establishing an algebraic bound applicable for any
monotone submodular function and any solution set of elements,
attained by an algorithm that considers the elements in a greedy
fashion. Note that our algorithm indeed considers the elements in
such fashion. We define the \emph{greedy elements sequence}
$\mathcal{E}(f,S) = \langle e_1, \ldots, e_n \rangle$ of a
submodular function $f$ and a set $S$ as the ordered sequence of
elements considered by a greedy process whose outcome is $S$.
Specifically, the greedy process is initialized with $R_0 = [n]$
and $S_0 = \emptyset$. Then , it runs for $n$ steps, where in each
step $t$, it considers the element $e_t \in R_{t-1}$ that has a
maximum marginal value with respect to the current solution set
$S_{t-1}$, and adds it to the solution set $S_t$ of the next step
if $e_t \in S$. In any case, the element $e_t$ is removed from
$R_{t-1}$ to obtain the set $R_t$ of remaining elements for the
next step. With this definition in mind, let $E_t =
\{e_1,\ldots,e_t\}$ be the set of first $t$ elements in the
sequence under consideration.

\begin{claim} \label{claim:SparseGreedyBound}
Given a submodular function $f: 2^{[n]} \to \bbR_+$, a set $S
\subseteq [n]$, their greedy elements sequence $\mathcal{E}(f,S) =
\langle e_1, \ldots, e_n \rangle$, and another set $S^* \subseteq
[n]$ satisfying $|S \cap E_t| \geq \alpha \cdot |S^* \cap E_t|$
for every $t \in [n]$ and a parameter $\alpha \leq 1$, it holds
that $f(S) \geq (\alpha / (\alpha+1)) \cdot f(S^*)$.
\end{claim}
\begin{proof}
Let us assume without loss of generality that the greedy process
goes over the elements according to the order $1$ to $n$, namely,
$E_1 = \{1\}, E_2 = \{1,2\}$, and so on. We note that this
assumption is valid since one can appropriately rename the
elements. Furthermore, let $S = \{a_1,\ldots,a_{|S|}\}$ and $S^* =
\{b_1,\ldots,b_{|S^*|}\}$ be the respective elements of $S$ and
$S^*$ sorted in an increasing order. Let us suppose that
$1/\alpha$ is integral. We emphasize that this assumption is
merely for simplicity of presentation, as we demonstrate later. We
match between each element of $S$ and $1/\alpha$ distinct elements
from $S^*$. Specifically, each element $a_t$ is matched to the
elements set $S^*_t = \{b_{(t-1)/\alpha+1},\ldots,b_{t/\alpha}\}$.
Notice that every element of $S^*$ is matched to an element of
$S$; else, it must be that $|S^*| > |S|/\alpha$, but this
contradicts the fact that $|S| = |S \cap E_n| \geq \alpha \cdot
|S^* \cap E_n| = \alpha |S^*|$. We next argue that each $a_t \leq
b_{(t-1)/\alpha+1}$. As a result, we attain that each
$$
f_{S \cap E_{a_t-1}}(a_t) \geq f_{S \cap
E_{a_t-1}}(b_{(t-1)/\alpha+1}), \ldots, f_{S \cap
E_{a_t-1}}(b_{t/\alpha}) \ .
$$
The last inequality holds since we known that when the element
$a_t$ was considered by the greedy process, all the elements of
$S^*_t$ were still available, and therefore, their marginal value
with respect to the solution $S \cap E_{a_t-1}$ was no more than
the marginal value of the element $a_t$. Consequently,
\begin{eqnarray*}
f(S^*) \leq f(S) + \sum_{b \in S^* \setminus S} f_S(b) & = & f(S)
+ \sum_{t = 1}^{\lceil \alpha |S^*| \rceil} \sum_{b \in S^*_t} f_S(b)\\
& \leq & f(S) + \frac{1}{\alpha} \sum_{t = 1}^{|S|} f_{S \cap
E_{a_t-1}}(a_t) = \left(1+\frac{1}{\alpha}\right) f(S) \ ,
\end{eqnarray*}
where both inequalities hold by the submodularity of $f$. For the
purpose of establishing the previously mentioned argument, suppose
by way of contradicting that there is some $t$ for which $a_t >
b_{(t-1)/\alpha+1}$. Let us concentrate on the elements set
$E_{(t-1)/\alpha+1}$. Notice that $|S \cap E_{(t-1)/\alpha+1}|
\leq t-1$, whereas $|S^* \cap E_{(t-1)/\alpha+1}| =
(t-1)/\alpha+1$. This implies that $|S \cap E_{(t-1)/\alpha+1}| <
\alpha \cdot |S^* \cap E_{(t-1)/\alpha+1}|$, a contradiction. We
conclude by noting that our assumption that $1/\alpha$ is integral
can be easily neglected. Specifically, one need to modify that
proof in such a way that a fractional part of an element from
$S^*$ may be matched to an element form $S$. Then, notice that at
most two fractional parts of an element of $T$ are matched to
elements of $S$, and those elements must appear before the element
of $S^*$ in the greedy elements sequence.~
\end{proof}

We now turn to establish our main result for the special setting
of maximizing a monotone submodular function under $k$-column
sparse packing constraints.

\begin{proofof}{Theorem~\ref{th:MainResult4}}
We first claim that the algorithm outputs a feasible solution,
namely, a solution that respects the packing constraints. Suppose
by way of contradiction that $\ell$ is the first element that is
added to $S$ and induces a violation in some constraint $i$ at
iteration $t$ of the main loop. Note that $A_{i \ell} = 1$. Let
$S_t$ be the solution at the end of iteration $t$, and notice that
$\sum_{j \in S_t} A_{i j} = b_i + 1$ since all the entries of $A$
are binary. This implies that $w_i = \lambda - 1$ at the beginning
of the iteration in which $\ell$ was considered, and thus,
$\sum_{i=1}^m A_{i \ell} w_i \geq \lambda - 1$. Inspecting the
selection rule, one can infer that $\ell$ could not have been
selected.

We next demonstrate that the algorithm attains an approximation
guarantee of $\Omega(1 / (Wk^{1/W}))$ when the update factor is
$\lambda = k + 1$. Recall that $W$ is the width of the
constraints, which is equal to $\min\{b_i\}$ in our case.
Similarly to before, we denote by $S^* \subseteq [n]$ a solution
that maximizes the submodular function subject to the linear
packing constraints. Let $\langle e_1, \ldots, e_n \rangle$ be the
ordered sequence of elements considered by our algorithm, and note
that it is essentially the greedy elements sequence
$\mathcal{E}(f,S)$. Moreover, let $E_t = \{e_1,\ldots,e_t\}$ be
the set of first $t$ elements in that sequence, $S^*_t = S^* \cap
E_t$ be the elements of $E_t$ in the optimal solution, $S_t = S
\cap E_t$ be the elements of $E_t$ in our algorithm's solution,
and $w_{it} = \lambda^{\sum_{j \in S_t} A_{ij} / b_i} - 1$ be the
value of $w_i$ at the end of iteration $t$ of the algorithm. We
prove the two following claims:

\begin{claim}
For every $t \in \{0, \ldots, n\}$,
$$
|S_t| \geq \frac{\sum_{i=1}^{m} b_i w_{it}}{W \lambda^{1/W} (k +
\lambda - 1)} \ .
$$
\end{claim}
\begin{proof}
We prove this claim by induction on $t$. The induction base is
when the algorithm begins, i.e., when $t = 0$. It is easy to see
that both sides of the above expression are zero in this case. In
particular, notice that all the weights are initialized to $0$.
Observe that in order to establish the induction step, it is
sufficient to demonstrate that when an element $\ell$ is selected
at iteration $t+1$ then $1 \geq \sum_{i=1}^{m} b_i \cdot
(w_{i(t+1)} - w_{it}) / (W \lambda^{1/W} (k + \lambda - 1))$. For
this purpose, notice that
$$
w_{i(t+1)} - w_{it} = \lambda^{\sum_{j \in S_t} A_{ij} / b_i}
\cdot \left(\lambda^{\left(\sum_{j \in S_{t+1}} A_{ij} - \sum_{j
\in S_t} A_{ij}\right) / b_i} - 1 \right) \leq \lambda^{\sum_{j
\in S_t} A_{ij} / b_i} \cdot \frac{W \lambda^{1/W}
A_{i\ell}}{b_{i}} \ ,
$$
where the inequality follows by plugging $a = \lambda^{1/W}$ and
$y = W / b_i \cdot (\sum_{j \in S_{t+1}} A_{ij} - \sum_{j \in S_t}
A_{ij}) = W A_{i\ell} / b_i$ to the inequality $a^y -1 \leq ay$,
which is known to be valid for any $a \in \bbR_+$ and $y \in
[0,1]$. As a consequence, we attain that
\begin{eqnarray*}
\sum_{i=1}^{m} b_i \cdot (w_{i(t+1)} - w_{it}) & \leq & W \lambda^{1/W} \sum_{i=1}^{m} A_{i \ell} \cdot  \lambda^{\sum_{j \in S_t} A_{ij} /b_i} \\
& = & W \lambda^{1/W} \sum_{i=1}^{m} A_{i \ell} \cdot \left( (\lambda^{\sum_{j \in S_t} A_{ij} / b_i} - 1) + 1 \right)\\
& = & W \lambda^{1/W} \left(\sum_{i=1}^{m} A_{i \ell} w_{it} + \sum_{i=1}^{m} A_{i \ell}\right) \\
& < & W \lambda^{1/W} \left((\lambda - 1) + k\right) \ ,
\end{eqnarray*}
where the last inequality holds since we know that (1) element
$\ell$ is selected at iteration $t+1$, and thus, $\sum_{i=1}^{m}
A_{i \ell} w_{it} < \lambda - 1$, and (2) the packing constraints
are $k$-column sparse, namely, the number of $1$-value entries in
each column is at most $k$, and hence, $\sum_{i=1}^{m} A_{i \ell}
\leq k$.
\end{proof}

\begin{claim}
For every $t \in \{0, \ldots, n\}$,
$$
|S^*_t| \leq |S_t| + \frac{\sum_{i=1}^{m} b_i w_{it}}{\lambda - 1}
\ .
$$
\end{claim}
\begin{proof}
Clearly, $|S^*_t| \leq |S_t| + |S^*_t \setminus S_t|$. Now, notice
that every element $j \in S^*_t \setminus S_t$ was not selected by
our algorithm when it was considered in step $t'+1$ since
$\sum_{i=1}^{m} A_{i j} w_{i t'} \geq \lambda - 1$. since the
weights may only increase during the run of the algorithm, we can
infer that
$$
(\lambda - 1) \cdot |S^*_t \setminus S_t| \leq \sum_{j \in S^*_t
\setminus S_t} \sum_{i=1}^{m} A_{i j} w_{i t} = \sum_{i=1}^{m}
w_{i t} \sum_{j \in S^*_t \setminus S_t} A_{i j} \leq
\sum_{i=1}^{m} w_{i t} b_i \ ,
$$
where the last inequality holds by recalling that the set $S^*_t
\setminus S_t$ is a subset of the optimal solution, and hence,
constitute a feasible solution respecting all constraints. As a
result, $\sum_{j \in S^*_t \setminus S_t} A_{ij} \leq b_i$.~
\end{proof}

We can now utilize the above claims and get that for every $t \in
\{0, \ldots, n\}$,
$$
|S^*_t| \leq |S_t| + \frac{\sum_{i=1}^{m} b_i w_{it}}{\lambda - 1}
\leq |S_t| + \frac{W \lambda^{1/W} (k + \lambda - 1)}{\lambda - 1}
|S_t| = \left(1 + 2W\lambda^{1/W}\right) \cdot |S_t| \ ,
$$
where the last equality holds as $\lambda = k+1$. Therefore, we
can employ Claim~\ref{claim:SparseGreedyBound} with $\alpha = 1 /
(1 + 2W\lambda^{1/W})$, and attain that the solution of our
algorithm approximates the optimal solution to within a factor of
at least $\alpha/ (\alpha + 1) = 1 / (2 + 2W\lambda^{1/W}) =
\Omega(1/(Wk^{1/W}))$.~
\end{proofof}

\paragraph{Acknowledgments:}
The authors thank Chandra Chekuri, Ilan Cohen, Gagan Goel, and Jan
Vondr{\'a}k for valuable discussions on topics related to the
subject of this study.


\appendix
\section{Submodular Maximization with Binary Packing Constraints} \label{appsec:0-1}
We study the special setting of monotone submodular maximization
under binary packing constraints, that is, when $A \in \{0,1\}^{m
\times n}$ instead of $A \in [0,1]^{m \times n}$. Note that we
assume without loss of generality that $b \in \bbN_+^m$. We
demonstrate that our multiplicative updates approach from
Section~\ref{sec:GeneralPacking} can be utilized to attain an
improved approximation guarantee for the underlying setting.
Specifically, we prove the following theorem.

\begin{theorem} \label{th:MainResult3}
There is a deterministic polynomial-time algorithm that achieves
an approximation guarantee of $\Omega(1 / m^{1/(W+1)})$ for
maximizing a monotone submodular function under binary packing
constraints.
\end{theorem}

Our approach for treating this case is identical to that of the
general case. We employ a multiplicative updates algorithm that is
identical to the algorithm presented for the general case with two
exceptions:
\begin{enumerate}
\item Line~\ref{alg:StopCond}: the first condition is changed to
$\sum_{i=1}^m b_i w_i < \lambda$ instead of $\sum_{i=1}^m b_i w_i
\leq \lambda$.

\item Line~\ref{alg:WeightUpdate}: the weights update is changed
to $w_i \leftarrow w_i \lambda^{A_{ij}/(b_i + 1)}$ instead of $w_i
\leftarrow w_i \lambda^{A_{ij}/b_i}$.
\end{enumerate}

We now prove that the modified algorithm for the binary case
outputs a feasible solution and attains the claimed approximation
ratio. Essentially, these results follow the analogous proofs of
the general case with some minor adjustments.

\begin{lemma} \label{lemma:0-1-Feasibility}
The modified algorithm outputs a feasible solution.
\end{lemma}
\begin{proof}
Let us focus on the solution $S$ when the main loop terminates.
Clearly, if $S$ respects the packing constraints then the returned
solution also respects them. Thus, let us consider the case that
$S$ is infeasible. We next argue that $S$ became infeasible only
at the last iteration of the loop in which element $\ell$ was
selected. Consequently, by inspecting the last two lines of the
algorithm, one can conclude that the returned solution must be
feasible.

For the purpose of establishing the previously mentioned argument,
let $\ell$ be the first element that induces a violation in some
constraint. Specifically, suppose $\ell$ induces a violation in
constraint $i$ at iteration $t$. This implies that $\sum_{j \in
S_t} A_{ij} = b_i + 1$ since all the entries of $A$ are binary.
Therefore,
$$
b_i w_{it} = b_i w_{i0} \prod_{j \in S_t} \lambda^{A_{ij} / (b_i +
1)} = \lambda^{\sum_{j \in S_t} A_{ij} / (b_i + 1)} = \lambda \ ,
$$
where the second equality is due to the fact that $w_{i0} = 1 /
b_i$. This implies that $\sum_{i=1}^m b_i w_{it} \geq \lambda$,
and hence, by inspecting the (modified) main loop stopping
condition, we know that the loop must have terminated immediately
after element $\ell$ was selected.~
\end{proof}

\begin{lemma} \label{lemma:0-1-Approx}
The modified algorithm archives $\Omega(1 /
m^{1/(W+1)})$-approximation by using $\lambda = e^{W+1} m$.
\end{lemma}
\begin{proof}
Suppose the main loop terminates after $t$ iterations. Notice that
when the loop terminates either $S_t = [n]$ or $\sum_{i=1}^m b_i
w_{it} \geq e^{W+1} m$. One can easily demonstrate that in the
former case, and in fact whenever $f(S_t) \geq f(S^*)$, the
returned solution is $1/2$-approximation to the optimal one. Thus,
in the remainder of the proof, we shall assume that $f(S^*)
> f(S_t)$ and that the loop terminates with $\Lambda_t =
\sum_{i=1}^m b_i w_{it} \geq e^{W+1} m$.

We concentrate on upper bounding the value of $\Lambda_t$. For
this purpose, we analyze the change in $\sum_{i=1}^m b_i w_i$
along the loop iterations. Observe that for any $\ell = 1, \ldots,
t$,
\begin{eqnarray*}
\Lambda_{\ell} = \sum_{i=1}^m b_i w_{i\ell} & = & \sum_{i=1}^m b_i w_{i(\ell-1)} \cdot \left(e^{W+1} m\right)^{A_{i\gamma(\ell)}/(b_i + 1)}\\
& \leq & \sum_{i=1}^m b_i w_{i(\ell-1)} \cdot \left(1 + \frac{(W+1) e m^{1/(W+1)} A_{i\gamma(\ell)}}{b_i + 1}\right)\\
& \leq & \sum_{i=1}^m b_i w_{i(\ell-1)} + (W+1) e m^{1/(W+1)} \sum_{i=1}^m A_{i\gamma(\ell)} w_{i(\ell-1)}\\
& = & \Lambda_{\ell - 1} + (W+1) e m^{1/(W+1)} \alpha_\ell
\delta_\ell \ .
\end{eqnarray*}
The first inequality can be obtained by plugging $a = e
m^{1/(W+1)}$ and $y = (W+1) A_{i\gamma(\ell)}/ (b_i + 1)$ to the
inequality $a^y \leq 1 + ay$, which is known to be valid for any
$a \in \bbR_+$ and $y \in [0,1]$, while the last equality results
from the definition of $\alpha_\ell$. By
Claim~\ref{claim:OPTBound}, we know that $\alpha_\ell \leq
\Lambda_{\ell-1} / (f(S^*) - f(S_{\ell-1}))$ in case $f(S^*) >
f(S_{\ell-1})$. The latter condition clearly holds since $f(S^*) >
f(S_t)$, and $f(S_t) \geq f(S_{\ell-1})$ for any $\ell$ under
consideration. Therefore,
$$
\Lambda_\ell \leq \Lambda_{\ell-1} \cdot \left(1 + \frac{(W+1) e
m^{1/(W+1)} \delta_\ell}{f(S^*) - f(S_{\ell-1})}\right) \leq
\Lambda_{\ell-1} \cdot \exp\left(\frac{(W+1) e m^{1/(W+1)}
\delta_\ell}{f(S^*) - f(S_{\ell-1})}\right) \ ,
$$
where the last inequality is due to the fact that $1 + y \leq
e^y$. The resulting recursive definition can be used, in
conjunction with the base case $\Lambda_0 = m$, to upper bound
$\Lambda_t$ by
$$
\Lambda_t \leq \Lambda_0 \cdot \prod_{\ell = 1}^{t} \exp
\left(\frac{(W+1) e m^{1/(W+1)} \delta_\ell}{f(S^*) -
f(S_{\ell-1})}\right) = m \cdot \exp \left((W+1) e m^{1/(W+1)}
\sum_{\ell = 1}^{t}\frac{ f(S_\ell) - f(S_{\ell-1})}{f(S^*) -
f(S_{\ell-1})}\right) \ .
$$
Recall that we assumed that the loop terminated with $\Lambda_t
\geq e^{W+1} m$. This lower bound on $\Lambda_t$ can be utilized,
together with the upper bound on $\Lambda_t$, to yield
$$
1 \leq em^{1/(W+1)} \sum_{\ell = 1}^{t}\frac{ f(S_\ell) -
f(S_{\ell-1})}{f(S^*) - f(S_{\ell-1})} \leq em^{1/(W+1)}
\ln\left(\frac{f(S^*) - f(S_0)}{f(S^*) - f(S_t)} \right) \ ,
$$
where the last inequality is due to the
Claim~\ref{claim:SubmodularBound}. Noting that $f(S_0) = 0$, one
can use simple algebraic manipulations and obtain that $1 /
(em^{1/(W+1)} + 1) \leq f(S_t) / f(S^*)$.~
\end{proof}

\end{document}